\documentclass[sigconf,authorversion,nonacm]{acmart}

\AtBeginDocument{%
  \providecommand\BibTeX{{%
    \normalfont B\kern-0.5em{\scshape i\kern-0.25em b}\kern-0.8em\TeX}}}

\usepackage{natbib}
\usepackage{amsfonts}
\usepackage{subcaption}
\usepackage{graphicx}
\usepackage{amsmath}
\usepackage[english]{babel}
\usepackage{amsthm}
\usepackage{url}
\usepackage{xcolor}
\usepackage{booktabs}
\usepackage{enumitem}
\usepackage{multirow}
\usepackage{algorithm}
\usepackage{algpseudocode}
\usepackage[fulladjust]{marginnote}
\newtheorem{theorem}{Theorem}
\newtheorem{lemma}{Lemma}
\newtheorem{corollary}{Corollary}[theorem]

\DeclareMathOperator*{\argmax}{arg\,max}

\graphicspath{ {./figures/} }

\AtBeginDocument{
  \AtBeginShipoutNext{
    \AtBeginShipoutUpperLeft{
      \put(\dimexpr\paperwidth-20.5cm\relax,-1.7cm){
        \parbox[b]{20cm}{
          Preliminary Draft Under Review: Not for Distribution
        }
      }
    }
  }
}
\begin{document}

\title{Rational Silence and False Polarization: How Viewpoint Organizations and Recommender Systems Distort the Expression of Public Opinion}

 \author{Atrisha Sarkar}
 \affiliation{Schwartz Reisman Institute for Technology and Society \\ Vector Institute  \\ \institution{University of Toronto }
 \country{Canada}}
 \email{atrisha.sarkar@utoronto.ca}

 \author{Gillian K. Hadfield}
 \affiliation{Johns Hopkins School of Government and Policy and Whiting School of Engineering (Computer Science)\\ 
 \institution{Johns Hopkins University }
 \country{USA}}
 \email{ghadfield@jhu.edu}

\begin{abstract}
  
Social media platforms are one of the most important domains in which artificial intelligence (AI) has already transformed the nature of economic and social interaction. AI enables the massive scale and highly personalized nature of online information sharing that we now take for granted. Extensive attention has been devoted to the polarization that social media platforms appear to facilitate. However, a key implication of the transformation we are experiencing due to these AI-powered platforms has received much less attention: how platforms impact what observers of online discourse come to believe about community views. These observers include policymakers and legislators, who look to social media to gauge the prospects for policy and legislative change, as well as developers of AI models trained on large-scale internet data, whose outputs may similarly reflect a distorted view of public opinion. In this paper, we present a nested game-theoretic model to show how observed online opinion is produced by the interaction of the decisions made by users about whether and with what rhetorical intensity to share their opinions on a platform, the efforts of viewpoint organizations (such as traditional media and advocacy organizations) that seek to encourage or discourage opinion-sharing online, and the operation of AI-powered recommender systems controlled by social media platforms. We show that signals from ideological viewpoint organizations encourage an increase in rhetorical intensity, leading to the \emph{rational silence} of moderate users. This, in turn, creates a polarized impression of where average opinions lie. We also show that this observed polarization can also be amplified by recommender systems that, pursuant to a platform's incentive to maximize engagement, encourage the formation of viewpoint communities online that end up seeing a skewed sample of opinion. Unlike existing models, these well-known online phenomena are not here attributed to distortion in the formation of opinions nor to the seeking out of like-minded others, but rather to the interaction of the incentives of users, viewpoint organizations, and platforms implementing recommender systems. In addition to showing how these interactions can play out in simulations, we also identify practical strategies platforms can implement, such as reducing exposure to signals from ideological viewpoint organizations and a tailored approach to content moderation. 

\end{abstract}







\maketitle
\section{Introduction}

Artificial intelligence (AI) is transforming the way communities learn about the views of their citizens. Without AI, we could not sustain the huge volume of viewpoint exchange on social media. AI-based recommender systems allow platforms to autonomously curate what will be served to an individual from a massive dynamic inventory of content, using fine-grained information about each user to curate a micro-targeted pinhole view of the world that can differ dramatically from person to person. This new landscape has transformed the public sphere. Gone is the era of one-to-many broadcasts of viewpoints by a limited number of broadcasters and publishers, replaced by a complex and dynamic media ecosystem. This transformation, mediated by recommender systems and online platforms, has important implications for our political economy. To the extent that policymakers, legislators, and courts look to social media to gauge public opinion on contentious issues, they may well be misled about where true public opinion lies. They may well come to believe that views are more polarized and extreme than they really are, and respond accordingly. This distortion in beliefs about public opinion then has the potential to distort policy and collective choice. Moreover, these dynamics also have implications for AI model training. Since large-scale internet data is commonly used to train AI models, any divergence between online and true public opinion could result in model behavior that is similarly unrepresentative. Thus, we need more robust theoretical frameworks to understand why and how opinions expressed in the digital public sphere may differ from representative public opinion. \par
The phenomena of the digital public sphere have attracted a great deal of attention from researchers across various disciplines, including economics, computer science, psychology, sociology, and political science \cite{nyhan2023like}. Among these, polarization has become one of the most discussed topics in both academic and popular discourse \cite{lorenz2023systematic, arora2022polarization}. Social media platforms that gained early prominence were applications that algorithmically connected users to friends, family, and strangers, and formed virtual communities \citep{aichner2021twenty}. For the most part, the analysis of polarization was based on the formation of echo chambers through these communities, patterns of online interactions that reinforce users’ beliefs and opinions by isolating them from opposing points of view \cite{golub2012homophily, garimella2018political}. However, simplistic attempts to mitigate this problem by exposing people to opposing viewpoints have been shown to make the problem even worse by increasing polarization \cite{bail2018exposure}. Alternate models have shown that polarization, especially the kind in which people form contradictory beliefs when presented with the same factual information, can be a rational phenomenon based on differences in subjective past experiences \citep{haghtalab2021belief, singer2019rational}. And a large-scale experiment on the Facebook platform demonstrated that increasing the share of a user's feed that comes from like-minded people and publishers did not significantly affect any measures of polarization \citep{nyhan2023like}. Overall, we now view polarization on online social platforms as a far more complex phenomenon than initially hypothesized. \par

The shift toward a more nuanced understanding of polarization has taken several forms. First, there is an increasing focus away from issue-based polarization to affective polarization; the former refers to a shift over time in individual opinions to more extreme views, whereas the latter refers to deep-seated partisan animosity and increased hateful rhetoric toward opposing viewpoints \citep{iyengar2019origins}. There is a rise in affective polarization not only in the political sphere but also in more general social issue discourse, and not only in the United States but worldwide \citep{boxell2022cross, yarchi2021political}. Second, related to partisanship, there is growing attention on the role of key political entities on online platforms, such as partisan news media organizations and political influencers, who can broadcast perspectives to the public, thus exacerbating affective polarization. For example, between 2016 and 2020, influencers in the United States became both more political and more polarized. Similar trends have been observed in other countries, such as India \citep{dash2022divided}. Given the well-documented impact of partisan media in traditional news on electoral outcomes \citep{dellavigna2007fox}, the influence of ideological media organizations on online platforms is an increasing concern. Third, there has been a shift from viewing polarization as solely an individual phenomenon to recognizing it as a population-level dynamic. While the former focuses on the factors that might cause an individual to shift their position to more extreme views, the latter examines the mechanisms of polarization operating at the population level. This broader perspective helps uncover phenomena such as false polarization or the perception gap, where differences in who expresses opinions versus who remains silent create the impression that the population is more polarized than it actually is. In fact, the more partisan views are, the greater the misperception about the opinions of the out-group \citep{levendusky2016mis, yudkin2019perception}. Reducing this misperception has shown promise in reducing affective polarization \citep{lees2020inaccurate}. Social recommendation systems are at the heart of these issues and play a key role in matching users with content to increase platform participation and engagement. As our understanding of these complex social processes evolves, we need improved formal frameworks to model social dynamics, from individual to population-level effects, and to develop strategies for better online platform design. \par

In this paper, we present a formal framework that models the interplay between these dynamics at the levels of individual users on platforms, various media organizations, and social recommender systems. From the perspective of individual users, we develop a game-theoretic model in which users strategically decide on the rhetorical intensity with which they will express their opinions. This decision is based not just on their personal opinion on a particular issue but also on their beliefs about the opinions and the rhetorical intensity they expect from out-group and in-group members in the community. We specifically model a user's costs and benefits of other's rhetoric: intense rhetoric from an in-group member boosts the utility a user expects from expressing their own opinion while intense rhetoric from an out-group member dampens own utility. Using this model, we show that, in equilibrium, people with more extreme opinions are more likely to use intense rhetoric, while those with moderate views are more likely to remain silent, a phenomenon we refer to as \emph{rational silence}. One implication of this model is that even when individuals' opinions do not change over time, the population can appear more polarized due to the suppressive effect increasing rhetorical intensity has on moderate opinion holders, leading to false polarization.
At the organizational level, we integrate the individual opinion expression model into a model of \emph{viewpoint stewarding} -- a process by which media organizations deliberately shape long-term beliefs about out-group and in-group opinions within a community. We refer to media organizations that engage in this process as viewpoint organizations. We consider two types of viewpoint organizations: participatory organizations, which aim to maximize the percentage of users who express their opinions on a platform, and ideological organizations, which aim to shift public opinion to their preferred position and distort the perceived average opinion. We elucidate the optimal strategies for each type of organization and show that ideological organizations achieve their goals by making a community believe that the out-group is more extreme than it actually is.
Finally, from the perspective of a platform’s social recommender system, which seeks to maximize individual engagement, the platform adapts to users by selecting content from either ideological or participatory organizations. We show that this dynamic results in the stratification of the population into distinct organizational communities with different characteristics. Participatory communities have a greater representation of moderate opinion holders with moderate second-order beliefs about out-group and in-group opinions. In contrast, ideological communities are characterized by more extreme opinions and second-order beliefs, while the silent population tends to have moderate own opinions but extreme second-order beliefs about the opinions of others. \par

The rest of the paper is organized as follows. In Sec. \ref{sec:normative_coordination_game}, we begin with the model of the game that captures individual opinion expression and rational silence. In Sec. \ref{sec:inst_stewarding}, we develop a model of viewpoint steering involving a single organization within a population. In this section, we derive the optimal signaling policies for the two types of viewpoint organizations, participatory and ideological. In Sec. \ref{sec:multi_inst}, we extend this framework to an environment with multiple viewpoint organizations, a population with heterogeneous beliefs, and a setting where the recommender system of an online platform mediates the organizational stewarding process. We demonstrate the formation of organizational communities and discuss the opinion and belief characteristics of each community. Finally, we conclude by presenting policy implications for mitigating the population polarization identified in our model.

\section{Related Work}
Polarization as a phenomenon of politics of the 21st century has been one of the most widely discussed topics in the academic literature in the fields of economics \citep{levy2019echo, boxell2017internet}, psychology \citep{jung2019multidisciplinary}, political science \cite{hare2014polarization}, computer science \citep{lim2022opinion}, communication studies \citep{kubin2021role}, and law \citep{fagan2017systemic}. The literature is vast, and in this section, we present literature that is closest to the questions addressed in this paper, namely, models of individual-level and affective polarization, models of self-censorship and silence, and the role of media and recommendation algorithms in online social systems. For a more general coverage of social media and polarization, we refer to existing systematic reviews of the literature \citep{arora2022polarization, tucker2018social, bramson2017understanding, kubin2021role, van2021social, kubin2021role}. 

\subsubsection*{Models of individual-level polarization} There is a well-established body of literature focused on empirical studies identifying polarization on social networks \citep{bakshy2015exposure, cinelli2021echo, tucker2018social}. However, understanding the mechanisms behind the complex social phenomenon of polarization of the public sphere through computational models is essential to explore interventions to mitigate the problem. DeGroote \citep{degroot1974reaching} provided one of the earliest models of how the opinion of others influences an individual's opinion, and Golub and Jackson \citep{golub2012homophily} applied a similar model to demonstrate how homophilic networks (a preference for associating with individuals who share similar opinions or beliefs) lead to belief and opinion polarization. However, the main focus of opinion dynamics models (see \citep{xia2011opinion} and \citep{peralta2022opinion} for an overview) is on how a network and interaction structure affect opinion change over time. In our model, the opinion of an individual stays fixed, and only the observation of the publicly expressed opinion at the population level changes over time. Polarization resulting from opinion dynamics over a network is mediated through the formation of filter bubbles \citep{pariser2011filter} and echo chambers, and consequently, concerns about the rise of polarization in broader public discourse have led to empirical approaches to their detection in online social systems \citep{nguyen2014exploring, bakshy2015exposure, alatawi2021survey}, and proposed models to break that effect \citep{helberger2018exposure, li2023breaking}. However, studies such as those of Bail et al. \cite{bail2018exposure} also show that presenting people with information that they perceive to contradict their deeply held beliefs can result in further polarization. \par 
Given the mixed evidence on interventions to reduce polarization based on filter bubbles, recent literature has followed two different paths of analysis of polarization. One looks at better mechanistic models to capture individual-level polarization, specifically, the question of why presenting people with the same information can lead to two separate conclusions. Within this branch of models, some are based on a cognitive process of motivated reasoning \citep{jost2022cognitive}, and others are based on a rational choice-based model \citep{dorst2023rational}. The explanation of polarization through rational choice-based models predates social media; for example, Sunstein \citep{sunstein1999law} identifies social comparison (desire for favorable perception by one's ingroup) and subjective persuasion as two mechanisms that can lead to polarization. Recent literature often relies on differences across individuals based on their past experiences, and examples of specific modeling approaches include an agent-based approach \cite{singer2019rational}, Bayesian approach \citep{jern2014belief}, and learning theoretic approach \citep{haghtalab2021belief}. In our model, although there is a similarity in the modeling paradigm of rational choice and its connection to polarization, the process that leads to polarization in our model is fundamentally different.  Whereas the above literature provides a rational explanation of polarization through individuals coming to different conclusions, changing opinions, or rejecting factual information, our model identifies an individual's strategic decisions whether and how intensely to express their opinion as a key factor contributing to polarization. 

\subsubsection*{Models of affective polarization} A second branch of the literature moves the focus from opinion polarization at the individual level to affective polarization \citep{iyengar2019origins} where the focus is on the extend of out-group animosity.  The dominant narrative around the rise of affective polarization appeals to social identity theory \citep{iyengar2019origins}. The theory suggests that individuals increasingly perceive their primary identity along partisan lines, and the increase in affective polarization is a consequence of a cluster of cognitive processes rooted in group identity-based social psychology. A smaller line of work uses rational choice-based models of affective polarization seek to provide a theoretical explanation for \emph{why} such a phenomenon occurs.  For example, based on national election data from four countries, Algara and Zur \citep{algara2023downsian} show that a Downsian model \citep{downs1957economic} of strategic behavior of voters based on ideological positions explains affective polarization more than partisan identification. As another example, Yaouanq \cite{le2018model} models ideological disagreement as arising from rational choice about how to interpret ambiguous evidence in the context of  motivated beliefs based on preferred policy outcomes. \par

\subsubsection*{Models of self-censorship} The \emph{spiral of silence} \citep{noelle1974spiral} is a widely discussed theory in communication studies in which minority opinion holders' fear of isolation drives them into self-reinforcing silence, leading to an absence of minority views from the population. Since the original spiral of silence model predates social media, this framework has seen growing interest in the context of online opinion expression, specifically on the question of whether the model can explain self-censorship in social media. As an answer, a meta-analysis of sixty-six studies by Matthes et al. \cite{matthes2018spiral} shows that a spiral-of-silence-like effect linking perceptions about opinions and the willingness to express own opinions can be  established empirically in online social systems, too. A closely related work on the computational modeling of dynamics leading to a spiral of silence is by Gaisbauer et al. \cite{gaisbauer2020dynamics}. Similar to our work, Gaisbauer et al. build a microfoundational account of opinion expression based on a game-theoretic model with incentives to stay silent or express opinion conditioned on the expression of other agents. However, there are a few key differences in the modeling constructs between our model and the Gaisbauer et al. model. First, the Gaisbauer et al. model uses a network structure and intensity of connection between agents on that network as the main factor that determines agents' perception of others' opinions. In comparison, viewpoint organizations play that role in our model, which helps us analyze a media organization's optimal strategies and how such organizations can distort beliefs. The second difference is that in Gaisbauer et al. \cite{gaisbauer2020dynamics}, opinion expression is discrete, that is, individuals either choose to express or stay silent, whereas, in our model, we capture this process through a choice of rhetorical intensity on a continuous spectrum. This enables us to establish a relation between an ideological opinion and the rhetorical intensity, both of which can lie on a spectrum. Our approach also connects the spiral of silence phenomenon directly to affective polarization.\par

\subsubsection*{Media and online platforms} Media and its role in polarization has been well studied in the economics and political science literature \citep{prior2013media}. The implications of the interaction between users' strategic information sharing and social media platform incentives have been developed in the context of misinformation sharing and fake news \citep{acemoglu2023model, papanastasiou2020fake, hsu2020news}. Although there are some similarities to our work, particularly in the use of game-theoretic models of interaction, such as Acemoglu's \citep{acemoglu2023model} model, which frames user content sharing as a game of strategic complements, there are important qualitative differences. Unlike the focus on information reliability and misinformation, our model examines the interplay between private opinions, rhetorical choices, and their impact on publicly observed opinions.\par

\subsubsection*{Social recommender systems and polarization} Approaches that explicitly model the algorithmic effect of social recommender systems often focus on the network topology, mainly how these systems help form homophilic networks by recommending connections between individuals with similar opinions. Musco et al. \citep{musco2018minimizing} propose a novel metric for recommender systems aimed at balancing the reduction of polarization (by connecting users with dissimilar opinions) without increasing the risk of disengagement due to disagreement. While homophilic community formation is one analysis, another set of models looks at how these connections influence opinion dynamics over time. For example, Morales and Cointet \citep{ramaciotti2021auditing} have examined the interplay between recommendation systems and models of opinion dynamics. They combine network-based recommender systems with a DeGroot opinion dynamics model \citep{degroot1974reaching}. Using simulation data based on Twitter interactions among French Members of Parliament, Morales and Cointet find that different recommendation algorithms impact polarization differently. Specifically, they show that the Alternating Least Squares \citep{hu2008collaborative} and Bayesian Personalized Ranking \citep{rendle2012bpr} algorithms increase polarization, while the Logistic Matrix Factorization algorithm \citep{johnson2014logistic} reduces it. Adding to this line of approaches, Santos et al. \citep{santos2021link} demonstrate that structural similarity—recommendations based on common neighbors in a network—can also exacerbate polarization. While these existing models provide valuable insights into how recommender systems contribute to polarization through opinion similarity and structural dynamics, they do not account for the strategic behaviors of users and the fact that non-engagement of users with the recommender systems itself can have a negative externality of polarization. Our work advances the latter line of research by incorporating the strategic behavior of users that models silence and expression, as well as the nested incentives of both platforms and media organizations.  

\begin{figure}[t]
 \centering
         \includegraphics[width=0.25\textwidth]{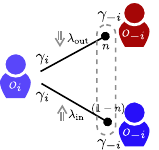}
\caption{Each player with a private opinion $o_{i}$ is matched with another player drawn from an opinion distribution $f$. With probability $n$ they are matched with someone from their in-group and probability $(1-n)$ from out-group. Each player strategically chooses their rhetorical intensity $\gamma_{i}, \gamma_{-i}$. If the player is matched with the in-group,  utility is boosted by a factor $\lambda_{\text{in}} > 1$ and, conversely,  utility is dampened by a factor $\lambda_{\text{out}} < 1$ when matched with an out-group member.}
\label{fig:game_interaction}
\end{figure}

\section{Opinion Expression Game and Rational Silence}\label{sec:normative_coordination_game}
In this section, we present the opinion expression game that models users' expression of opinion on a given issue. We assume a context in which there is a matter of public debate in front of the community, such as whether to adopt a legislator's proposed policy or how a supreme court should resolve a case it is scheduled to decide. Individuals hold a spectrum of opinions on this issue but can be grouped based on a binary partition: approval (the policy should be adopted, the court should be hold in favor of the plaintiff) and disapproval (the policy should be rejected, the court should hold in favor of the defendant). We model individuals' choices about whether to express their opinion on the matter or remain silent as an iterated game in which individuals in each iteration are randomly matched with someone else from the community. This is a stylized representation of online interactions where users engage with a social media platform periodically and primarily interact with other users who are on the platform at (roughly) the same time. In any iteration, an individual may find themselves matched with someone  in their \lq in-group\rq  who shares their binary position (approval or disapproval) on the question or \lq out-group\rq,  who holds the opposite position. First, we present the main constructs of the game, followed by an equilibrium analysis that determines individual behavior. \par
\subsection{Opinions} Each individual, indexed as $i$, holds private opinions $o_{i} \in [0,1]$ on a focal issue. Using game-theoretic terminology, the opinion of the player $i$ is their \emph{ type}, which is private information. The opinions are drawn from an arbitrary opinion distribution $f_{(o_{A},o_{D})}$, where $o_{A},o_{D}$ are the mean approval and disapproval opinions on which the distribution is parameterized. The values $o_{i} < 0.5$ indicate \textit{ disapproval} on the focal issue and $o_{i} \geq 0.5$ indicates \textit{approval}.  Opinions lying on the continuum between 0 and 1 represent the diversity of opinions that go beyond a simple yes/no dichotomy. At the two extremes, an opinion of $o_{i}=1$ or $o_{i}=0$ denotes complete support for or against the matter under debate, respectively: support for all elements of a proposed policy, for example, or opposition to a finding for the plaintiff on any possible grounds in a court case. We then conceptualize an opinion between 0 and 1 as an indication  that on average an individual may support (or not support) a particular outcome, but were a narrower question posed, their position might change. For example, suppose the matter under debate is whether animal testing of products should be prohibited. Someone might support a proposed total ban even though they would prefer that testing be allowed in a narrow set of cases for the development of life-saving drugs. The opinion $o_{i}$ can be thought of as capturing the strength of approval or disapproval with respect to a particular contentious question. Note that unlike most work in the literature on social media, we assume people's opinions on a matter are fixed and not subject to influence or change. \par
\begin{table*}[htbp]

\begin{tabular}{@{}c|c|p{12.5cm}}
\toprule
Parameters& Parameter type   &  Description\\ \midrule
$\hat{n}, o_{A}, o_{D}$   & \multirow{1}{*}{Descriptive belief} & Estimated proportion of the population that holds approval opinions, the estimate of the mean approval and disapproval opinion, respectively.\\
\midrule
$\gamma_{-i}$& \multirow{2}{*}{Strategic choices} & Rhetorical intensity of the non-focal player\\
$\gamma_{i}$& & Rhetorical intensity of the focal player\\ \midrule
$o_{i}$& Private information & Opinion of the focal player\\
\midrule
$\alpha$&    & Cost of using rhetoric to express opinions, including possible penalties imposed by platform's moderation policy\\ 
$\lambda_{\text{in}}$&  Exogenous parameters  & Utility boosting factor from agreeing with in-group\\ 
$\lambda_{\text{out}}$&    & Utility reduction factor from disagreement with out-group\\ 
\bottomrule
\end{tabular}

\caption{Parameters that define the interaction between two players within a group.}
\label{tab:params}
\end{table*}
\subsection{Rhetorical intensity} Whereas opinion represents an individual private position on an issue, \textit{rhetorical intensity} refers to \emph{how} one chooses to express it. Continuing on the previous example, an individual who has absolute support for a total ban on animal testing ($o_{i}=1$) and who chooses to express this opinion can choose to express themselves mildly, with a simple statement of their view. Or they could choose to express their opinion in a highly partisan way with verbal abuse of those holding an opposing view. We model this choice as a continuous variable $\gamma_{i}$ $\in [0, 1]$. Values close to $0$ represent remaining silent, low values represent a mild expression of one's opinion, and values close to $1$ represent a threat of violence or other forms of extreme behavior.
\begin{figure}[t]
    \centering
    \includegraphics[width=0.3\textwidth]{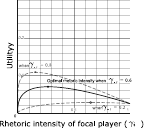}
    \caption{The net utility of expressing an opinion (after deducting the cost) as a function of rhetorical intensity for the focal player. The utility-maximizing rhetorical intensity of the focal player decreases with an increase in other players' rhetorical intensity.}
    \label{fig:utility_wrt_phi}
\end{figure} 
\subsection{Cost and utilities} The opinion expression game is played between an individual, indexed as $i$, and a randomly selected member of the community, indexed as $-i$. For purposes of our simulations we construct a utility function that captures the following ideas. First, a person experiences utility directly from expressing their opinion and this scales with the strength of their agreement with the proposal under debate. So, a person who is one-hundred percent in agreement with the question ($o=1$) gets more utility from expression than someone who disagrees with some of the fine-grained elements of the matter under debate but is still overall supportive. (Using our earlier example, someone who supports a total ban on all animal testing gets more utility from expressing that view than someone who would prefer a policy that allowed exceptions for the development of life-saving drugs but will still choose to support rather than oppose the proposed total ban.) We also assume that a person's utility from expression is increasing in the level of rhetoric they use. 
Next, we take into account that the utility from expressing one's opinion depends also on whether one is matched with someone who is in overall agreement on a particular question (a member of one's in-group) or someone with whom one disagrees (a member of one's out-group.) Specifically, utility from opinion expression is increased (decreased) when one exchanges views with a member of the in-group (out-group). Finally, we assume that there is a positive externality associated with the rhetoric used by members of one's in-group. This affects an individual's own rhetoric level because we also assume that it is costly to use rhetoric, and increasingly so as the intensity of rhetoric increases. This implies that a user can enjoy the same level of utility from opinion expression at lower levels of own rhetorical intensity when interacting with someone who engages in more intense rhetoric to express the same (similar) viewpoint. We capture these ideas in the following utility function shown here for a focal individual with an approval opinion:
\small
    \begin{align}
    \begin{split}u_{i}(\gamma_{i},\gamma_{-i};o_{i}) = & \underbrace{ \hat{n} \cdot o_{i} \cdot \lambda_{\text{in.}} \cdot \gamma_{i}^{(1-\gamma_{-i})}}_{\text{\parbox{4cm}{\centering Utility from expression when matched with in-group member}}} \\ & \underbrace{+ (1-\hat{n}) \cdot  o_{i}  \cdot \lambda_{\text{out}}^{\gamma_{-i}} \cdot \gamma_{i}}_{\text{\parbox{4cm}{\centering Utility from expression when matched with out-group member}}}\\ &- \underbrace{\alpha\cdot \gamma_{i}}_{\text{\parbox{3cm}{\centering Cost of expressing opinion}}}
    \end{split}
    \label{eqn:utility}
    \end{align}
\normalsize
and for individuals with opinions of disapproval as:
\small
\begin{align}
    \begin{split}u_{i}(\gamma_{i},\gamma_{-i};o_{i}) = & \underbrace{ (1-\hat{n}) \cdot (1-o_{i}) \cdot \lambda_{\text{in.}} \cdot \gamma_{i}^{(1-\gamma_{-i})}}_{\text{\parbox{4cm}{\centering Utility from expression when matched with in-group member}}} \\ & \underbrace{+ \hat{n} \cdot  (1-o_{i})  \cdot \lambda_{\text{out}}^{\gamma_{-i}} \cdot \gamma_{i}}_{\text{\parbox{4cm}{\centering Utility from expression when matched with out-group member}}}\\ &- \underbrace{\alpha\cdot \gamma_{i}}_{\text{\parbox{3cm}{\centering Cost of expressing opinion}}}
    \end{split}
    \label{eqn:utility_disappr}
    \end{align}
\normalsize
where $\hat{n}$ is the proportion of the population with approval opinions. $\lambda_{in} > 1$ and $\lambda_{out} < 1$ are boosting and reducing constants that represent the increase and decrease of utility from being matched with the in-group or out-group, respectively. $\gamma_{i}$ and $\gamma_{-i}$ are the rhetorical intensity of expression of opinion for the focal and non-focal player, respectively. Note that we model out-group rhetorical intensity as a factor that increases the extent to which the utility enjoyed from own-opinion expression is dampened by being matched with someone with whom one disagrees. $\alpha$ is a cost factor that captures any costs incurred by  rhetorical intensity. This could include personal psychic or time costs as well as costs imposed by a social media platform such as a flag, demoting a post in a newsfeed, or limiting a user's access to the platform. We thus also allow $\alpha$ to be interpreted as a decision variable for the platform with a stricter moderation policy (higher $\alpha$) imposing a higher cost on the use of higher levels of rhetoric in expressing one's opinion. 
\begin{figure*}[t]
 \centering
         \includegraphics[width=\textwidth]{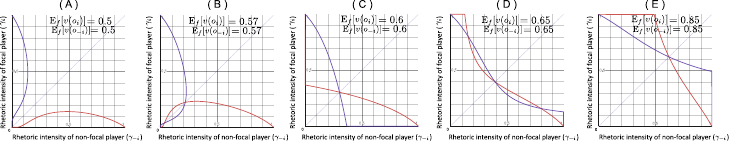}
\caption{Best response curves of the optimal rhetorical intensity for the focal (shown in red) and non-focal (shown in blue) player as a function of their opinion drawn from an opinion distribution of equal support for approval and disapproval. For the purpose of illustration, the constants $\alpha,\lambda_{\text{in}},\lambda_{\text{out}}$ are set to 0.7, 2, and 0.5, respectively. The intersection of the two curves represents the Bayesian Nash equilibrium rhetorical intensity in the \emph{ex-ante} form. The plots are shown for different opinions (types) of the focal and non-focal player, and we see that as the opinion moves to the extreme, the symmetric equilibrium rhetorical intensity becomes higher. }
\label{fig:br_curves}
\end{figure*}
\normalsize

\subsection{Optimal rhetoric} The utility function given in Equations \ref{eqn:utility} and \ref{eqn:utility_disappr} is concave in $\gamma$ for a focal individual and this means that we can solve for an optimal rhetorical intensity for this individual. Figure \ref{fig:utility_wrt_phi} shows this relationship. Rhetorical intensity of individual $i$ is shown on the $x$ axis and the $y$ axis represents total utility. The optimal rhetorical intensity also depends on the rhetorical intensity of the matched player's expression of opinion ($\gamma_{-i}$), shown in the figure as the three separate utility curves. The optimal rhetorical intensity of the focal player $i$ decreases with an increase in the other player $-i$'s rhetorical intensity. When matched with a member of the in-group, this effect arises because the other's rhetoric is a substitute for own rhetoric. When matched with a member of the out-group, this effect is because extreme rhetoric from the other side increases the extent to which utility from own expression is dampened.

\subsection{Equilibrium}
\label{sec:equilibrium}

We can now calculate the equilibrium rhetorical intensity that determines the outcome of the game. In order to facilitate our analysis, we focus on the symmetric Bayesian Nash equilibrium in continuous strategies. Within Bayesian Nash equilibrium, there are two possible forms of analysis; in the \emph{ex-ante} form, the analysis is based on the expected utility (over types) for each player, and in the \emph{ex-interim} form, each player is aware of their own type, and therefore, the players create their own view of the game in which their utility depends on their type and the utility of the other player is an expected utility over all possible types. In this section, we focus the analysis on the former form, and we rely on the latter form and the connection between the two for simulation in Sec. \ref{sec:long_run_dynamic}. \par 
Note that the opinion lies in the interval $[0,0.5)$ for the disapproval group and in the interval $[0.5,1]$ for the approval group. In order to make the value of expressing the opinion symmetric between the groups, we introduce the variable $v(o_{i})$ such that $v(o_{i}) = o_{i}$ when $o_{i} \geqslant 0.5$ and $(1-o_{i})$, otherwise. Next, we consider this value of the opinion ($v(o_{i})$) of each individual as their \emph{type} as in standard game-theoretic terminology. One can treat the type here as capturing the extremity of the opinion holder; moderate opinion holders will have a lower type, and extreme opinion holders will have a higher type. Finally, the formal form of the \emph{ex-ante} utility of a representative individual under equal support of approval and disapproval group size (i.e, $F(0.5) =0.5$ for opinion distribution $f$) is given by the following equation (Lemma \ref{lemma:exp_util}, \textit{c.f} Appendix):
\begin{align*}
\mathrm{E}_{f}[u_{i}(\gamma_{i},\gamma_{-i})] 
= &\hat{n} \cdot \mathrm{E}_{f}[v(o)] \cdot \lambda_{\text{in.}} \cdot \gamma_{i}^{(1-\gamma_{-i})} \\ &+ (1-\hat{n}) \cdot  \mathrm{E}_{f}[v(o)]  \cdot \lambda_{\text{out}}^{\gamma_{-i}} \cdot \gamma_{i} - \alpha\cdot \gamma_{i}
\end{align*}

The best response of a representative player $i$ to another representative player $-i$ can be found by taking the partial derivative of the above utility function with respect to $\gamma_{i}$. Similarly, the best response of a player $-i$ to a player $i$ can be found by taking the partial derivative with respect to $\gamma_{-i}$. This gives us the best response curves for \textit{i} and \textit{-i}:
\begin{equation}
    BR_{i}(\gamma_{-i};o_{i}) = \min \left( 1, \max \left( 0, \left( \frac{ \hat{n} \cdot \mathrm{E}_{f}[v(o_{i})] \cdot \lambda_{\text{in}} \cdot (1-\gamma_{-i})}{\alpha - (1-\hat{n})(\mathrm{E}_{f}[v(o_{i})]\lambda_{\text{out}}^{\gamma_{-i}})} \right)^{\frac{1}{\gamma_{-i,}}}\right) \right)
\label{eqn:br_eqns_a}
\end{equation}
\begin{equation}
    BR_{-i}(\gamma_{i};o_{-i}) = \min \left( 1, \max \left( 0, \left( \frac{ \hat{n} \cdot \mathrm{E}_{f}[v(o_{-i})] \cdot \lambda_{\text{in}} \cdot (1-\gamma_{i})}{\alpha - (1-\hat{n})(\mathrm{E}_{f}[v(o_{-i})]\lambda_{\text{out}}^{\gamma_{i}})} \right)^{\frac{1}{\gamma_{i,}}} \right) \right)
    \label{eqn:br_eqns_b}
\end{equation}

The Bayesian Nash equilibrium of an imperfect information game assigns strategies for every possible \emph{type} of a player. This means that in our case, there is an equilibrium rhetorical intensity as a function of the opinion value $v(o)$. We construct this function by solving the set of Equations \ref{eqn:br_eqns_a} and \ref{eqn:br_eqns_b} based on the symmetric equilibrium condition $BR_{i}(\gamma^{*}_{-i};o_{i}) = BR_{-i}(\gamma^{*}_{i};o_{-i})$ where $(\gamma^{*}_{i},\gamma^{*}_{-i})$ represents the pair of equilibrium rhetoric intensities. Other than the parameters $\hat{n}$ and $\alpha$, which are treated as constants for the purposes of solving for the equilibrium, the solutions to these equations are a function of the expected opinion value $v(o)$. \par 

Fig. \ref{fig:br_curves} shows the plot of the two best response curves (the focal player is shown in red, the best response curve of the non-focal player is shown in blue), the equilibria (intersection of the two curves), and how the equilibria change with increasing opinion value. Consider first optimal rhetorical intensity for the most moderate opinions ($\mathrm{E}_{f}[o_{i}] = 0.5$) (Panel A). For individuals with such opinions, if they anticipate engaging with someone who does not express an opinion ($\gamma_{-i} = 0$) their best response is to also remain silent ($\gamma_{i}$).  This is attributable to the low rewards for opinion expression for moderates. When the cost of expression (mediated by $\alpha$) is sufficiently high, these rewards do not warrant expression when there is no boost from being matched with someone who shares their views and whose rhetoric confers additional utility on the focal player. That boost occurs as the rhetorical intensity increases for the non-focal player, but it is counterbalanced by the possibility that the non-focal player is from the out-group and instead of a boost there is a dampening of the returns to opinion expression. We see that these combined effects initially lead the focal player to incur the cost of mild rhetoric but as the rhetorical intensity of the partner increases, the dampening effect comes to dominate and optimal rhetoric for the moderate focal player drops again, ultimately again inducing silence ($\gamma_{i} = 0$) as the partner's rhetoric reaches an extreme ($\gamma_{-i} = 1$).  As opinions become more extreme, however ($\mathrm{E}_{f}[v(o_{i})] = 0.6$, Panel C), the optimal response of a representative player to someone who is expected to remain silent is to use significant rhetoric to express themselves ($\gamma_i \approx 0.4$): there is no risk of bearing the cost of rhetoric from an out-group member. But as the rhetoric from the matched player increases, that cost reduces the expected return to the focal player's rhetoric and for a fixed cost of rhetoric, the optimum decreases. With very high rewards to opinion expression (Panel D), the focal player engages in maximal rhetoric until facing moderately intense rhetoric from the partner, at which point their own rhetoric begins to moderate. 

Focusing on the symmetric equilibria, i.e equilibria in which the optimal rhetorical intensity of the focal player and the non-focal player are the same ($\gamma^{*}_{i}=\gamma^{*}_{-i}$), we see that with increasing opinion values $o_{i}$, the equilibrium rhetorical intensity becomes higher. This is because at more extreme opinion values, the higher utility generated by more extreme opinions offsets the higher cost incurred from increased rhetorical intensity for both players. The following theorem captures these results.

\begin{theorem}
Let $v(o_{i})=o_{i}$ if $o_{i} \geqslant 0.5$ and $v(o_{i})=1-o_{i}$, if $o_{i} < 0.5$ and $o_{i} \sim f$ is drawn from any arbitrary prior opinion distribution with \texttt{p.d.f} $f$ and \texttt{c.d.f} F, and $\hat{v}_{o} = \mathrm{E}_{f}[v(o_{i})]$. Then, under the condition of equal support of approval and disapproval, i.e., $F(0.5)=0.5$, there exists an ex-ante symmetric Bayesian Nash equilibrium $\gamma_{i}^{*}=0$, $\forall i$, if and only if $\hat{v}_{o} < \frac{\alpha}{1-\hat{n}\cdot (1-\lambda_{in})}$  
\label{theo:exp_bound}
\end{theorem} 
\begin{proof}
In Appendix

\end{proof}
This theorem states that individuals with  opinions that generate utility below a  threshold, will choose to stay silent ($\gamma=0$). This threshold increases as the cost of expressing opinions ($\alpha$) increases, their belief about the fraction of the population that shares their opinion (their in-group) rises, and the boost they get from sharing opinions with their in-group ($\lambda _{in}$) shrinks. Conversely, those with more extreme views, above this threshold, will express their views. This leads to a predictable pattern of online opinion sharing being biased to those with more extreme views and is a theoretical demonstration of false polarization, with expressed views that are not representative of the true distribution of opinions in the population. 

 \begin{figure}[t]
 \centering
\includegraphics[width=0.4\textwidth]{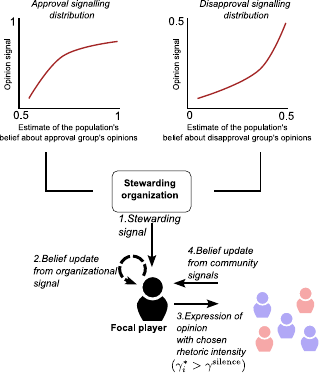}
\caption{Schematic representation of the organizational stewarding process. (1) An organization randomly samples a signal that conveys information about approval or disapproval group's opinions from the organization's signaling scheme. (2) The focal agent updates their descriptive belief based on the signal (3)). The focal agent estimates the equilibrium rhetorical intensity and expresses opinions. (4)) A community of agents also expresses opinions based on organizational signals. The focal agent and others updates descriptive beliefs based on the expressed opinion in the community.}
\label{fig:single_inst_stewarding_diag}
\end{figure}
\section{Organizational stewarding}
\label{sec:inst_stewarding}
The equilibrium analysis presented in the previous section is still based on a static game, meaning that we establish a one-shot interaction in which the expressed rhetorical intensity is a function of the opinions players expect others to hold and the rhetorical intensity with which those players are expected to express their views. In this section, we present a higher-level dynamic that models the belief stewarding process, whereby individuals choose actions based on the static opinion expression game, and \textit{viewpoint organizations} can influence the beliefs in the game by conveying strategically relevant information, that is, information about the in-group and out-group's opinions. We show that based on their own long-run objectives, viewpoint organizations can use their coordination capacity to sway publicly expressed opinions by shaping who expresses their opinion and who stays silent. \\

\noindent \textbf{In-group and out-group signaling}. The nature of the viewpoint organizations we consider in our model are those that convey information about the population's opinions. Traditional media, such as newsprint and talk radio, as well as new media organizations, such as political influencers and independent online media, fall into this category. We model the actions of these organizations as a set of signals, one conveying information about the population that approves a particular position, and another about the population that disapproves of that position. As a practical matter, these signals are generated by actions such as the choice made by an organization about who to interview for a particular story or whose opinion to share on social media, or what terms to use to describe a debate. These choices produce slant and bias \cite{mullainathan2005market} in the content of the reporting that conveys information about each of the positions of the group. \\

\noindent \textbf{Signaling policy and bounded confidence constraints}. By adding their own slant to a story that conveys information about a population's opinion, media organizations can generate biased signals that differ from the true opinion of the two groups \cite{rodrigo2023systematic}. We refer to this slant and bias as the organization's signaling policy and can be represented as the distribution $\pi(o_{I}|\hat{o})$, where $o_{I}$ is the organization's signal and $\hat{o}$ is the organization's estimate of the true approval or disapproval opinion. Even though an organization can bias its signaling policy, the generated signals cannot be arbitrary; rather, those signals need to convey some meaningful information about the population. Based on the Bounded Confidence model of Hegselmann and Krause \citep{rainer2002opinion}, we model this requirement with the help of a constraining distribution $\mathcal{C}$ that connects the generated signal and the estimated true belief. The constraining distribution $\mathcal{C}$ models the random variable of the absolute difference between the signal and the beliefs. We choose $\mathcal{C}$ as a uniform distribution in the domain $\hat{o} \pm \frac{\tau}{2}$ and 0 outside of this range. Although we choose a specific distribution for our analysis, without loss of generality, one can choose an alternate distribution that better reflects the relation between the organization's signaling constraint and the true beliefs. The only requirement in the organizational stewarding process is that the signal should convey \emph{some} information and not just be pure noise. \\

\subsection{Organizational stewarding sequence}
\label{sec:sub_inst_signal}
The organizational stewarding process involves repeated cycles of signal generation by the organization acting as sender and belief update by community members acting as receivers. Fig. \ref{fig:single_inst_stewarding_diag} shows the sequence of steps involved in one cycle, and we describe the sequence in more detail below.
\begin{enumerate}[wide, labelwidth=!, labelindent=0pt]
    \item \textit{Organizational signal generation (stewarding signal):} The organization holds prior belief about the mean opinion of the approval and disapproval group. It generates a signal about the approval and disapproval group's opinion from two separately chosen distributions. Each of these distributions is conditioned on the organization's estimate of the groups' true mean opinions. We elaborate upon the choice of the signal generating distribution from the organization's perspective in the next section (Sec. \ref{sec:opt_signalling}) 
    \item \textit{Belief update (from organizational signal):} When a player receives the signal, they interpret the signal as one about their \emph{in-group} or \emph{out-group} based on their own opinion $o_{i}$, and update their descriptive belief about the corresponding group based on Bayes rule as follows:
    \begin{equation}
        f_{\text{posterior}}(\hat{o}_{t+1}|o_{I,t}) = \frac{\mathcal{C}(o_{I,t};\hat{o}_{t},\tau) \cdot f_{\text{prior}}(\hat{o}_{t})}{\int \mathcal{C}(o_{I,t};\hat{o}_{t},\tau) \cdot f_{\text{prior}}(\hat{o}_{t}) \,d\hat{o}_{t}}
    \label{eqn:belief_update}
    \end{equation}
    where $\mathcal{C}(o_{I,t};\hat{o}_{t},\tau)$ is the likelihood of the organization generating the signal based on the constraining distribution $\mathcal{C}$ and $f_{\text{prior}}(\hat{o}_{t})$ is the prior belief about the corresponding group at time-step $t$. $f_{\text{posterior}}(\hat{o}_{t+1}|o_{I,t})$ is the posterior belief about the group after the receiver updates their belief based on the signal. The limits of the integral in the marginal are [0.5,1] or [0,0.5] depending on whether the signal is being generated for the approval or disapproval group, respectively.
    \item \textit{Opinion expression:} Based on the posterior beliefs about the opinions, players estimate the equilibrium rhetorical intensity of the group and then best respond correspondingly. We rely on the equivalence between \emph{ex-interim} and \emph{ex-ante} Bayesian Nash equilibrium to simulate this process \citep{fujiwara2015bayesian}. Players first estimate the symmetric \emph{ex-ante} equilibrium based on the belief about the mean opinions of the groups. Let $\gamma^{*}_{\text{\emph{ex-ante}}}$ denote that value. Subsequently, each player best responds following Eqn. \ref{eqn:br_eqns_a} as $BR_{i}(\gamma^{*}_{\text{\emph{ex-ante}}};o_{i})$. This response, which simulates the individual, is the \emph{ex-interim} response since own opinion (type) is known to the player. Next, individuals with rhetorical intensity less than a threshold $\gamma^{\text{silence}}$ stay silent, whereas others express their opinions.
    \item \textit{Belief update (from community signal):} Based on the mean disapproval and approval of the expressed opinions, everyone in the population (both those who expressed and who stayed silent) update their beliefs using Bayes rule. The update is similar to Eqn. \ref{eqn:belief_update}, but this time, these signals come from the community expressing their opinions.
\end{enumerate}

 \begin{figure*}[htbp]
 \centering
\includegraphics[width=0.9\textwidth]{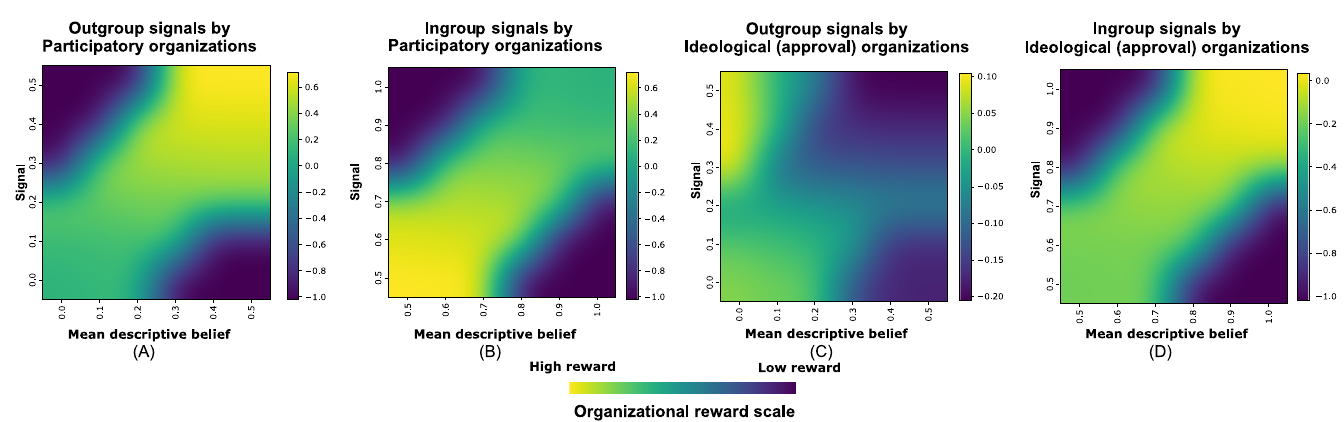}
\caption{Reward heatmap that shows the optimal signaling policy for different viewpoint organizations and for out-group and in-group signals. Brighter regions (yellow) represent higher reward. The $y$ axis shows the organizations' estimate of the true out-group/in-group belief and the $x$ axis shows the corresponding signals to generate. We see that the optimal signaling distribution follows different patterns for both organizations: ideological organizations get higher reward for more extreme signals about the groups' opinions whereas participatory organizations get higher reward from moderate signals.}
\label{fig:inst_rewards}
\end{figure*}

\subsection{Participatory and Ideological organizations}
\label{sec:opt_signalling}
In this section, we focus on step 1 of the organizational stewarding process, that is, the optimal signal generation on the part of the organizations. We analyze how different organizational incentives shape the signaling structure, and the effects of that on the beliefs and subsequent effect on the expression of opinion of the population. We consider two types of organizations: i) \emph{Participatory organizations:} organizations whose objective is to maximize the proportion of the population that express their opinion, and ii) \emph{ ideological organizations:} organizations whose objective is to move mean opinion towards one extreme, that is, 0 or 1. We show that each of these incentives results in a different signaling strategy by the two types of organizations. \par 

To compute the optimal signaling strategy, the organization needs to solve an optimization problem. The optimization problem can be formulated as a Markov Decision Process (MDP) with the organization's estimate of the mean descriptive belief of the approval and disapproval opinions as the state and the generated signal as the action; the optimal policy generates the mapping from one to the other. Note that the state transition, i.e. the change in the descriptive belief of the population from one time step to the next in the organizational stewarding process, depends only on the belief at the previous time step. A time step here refers to one cycle of the organizational stewarding process. Therefore, the Markovian nature of this transition makes it apt for the optimization problem from the perspective of the organization to be formulated as an MDP. For each organization, there are two separate MDPs to be solved; one that generates the optimal signaling policy for information about the approval group and one for the disapproval group. In both MDPs, the usual constructs of \textit{State, Action, Transition, Rewards} is as follows:
\begin{itemize}[leftmargin=*]
    \item \textit{State:} The descriptive belief about the population's approval or disapproval opinion, for each of the two MDPs, respectively. For simplicity, we can denote this as $\hat{o}_{t}$, but it takes values between [0,0.5) for disapproval and [0.5,1] for approval.
    \item \textit{Action:} The signals to generate. Similar to the state variable, this action, $o_{I,t}$, lies between [0,0.5) for disapproval and [0.5,1] for approval group signaling.
    \item \textit{Transition:} The change in the belief, which is based on the two sequential Bayesian updates in step 2 to step 4 applied in sequential manner as described in Sec. \ref{sec:sub_inst_signal}.
    \item \textit{Rewards:} We formulate two reward functions, one for participatory organizations ($R_{\text{pa}}$) and another for the ideological organizations ($R_{\text{id}}$). 
    \begin{align*}
    R_{\text{pa}}(\hat{o}_{t},o_{I,t
    }) &= 2(\frac{|N_{E}^{t+1}|}{N}-0.5) \\ R_{\text{id}}(\hat{o}_{t},o_{I,t},\hat{o}_{t+1}) &=4\bar{o}^{E}_{\geqslant 0.5} - 3
    \end{align*}
    where $\frac{|N_{E}^{t+1}|}{N}$, is the proportion of agents who express their opinion at time $t+1$. For the ideological objective, $\bar{o}^{E}_{\geqslant 0.5}$ is the mean opinion of the population that expresses approval; a value closer to 1 means higher the rewards for the ideological organization. The reward structure is also constructed in this way to make both of them bounded in the interval [-1,1]. 
\end{itemize}

\subsubsection{Optimal signaling policy}
\label{sec:optimal_signalling}
Based on the reward structure of the two organizations, we solve the MDP using the Value Iteration algorithm \citep{sutton2018reinforcement} after discretizing the state space. We choose this algorithm since the simplicity is sufficient for demonstration, however, for a more complex application in the real-world, such as when a media organization needs to determine the policy for their choice of content across different issues, a more sophisticated approach may be necessary. In Fig. \ref{fig:inst_rewards}, we plot the reward heatmap for the entire signaling space for both sets of signaling distributions (approval and disapproval). The $y$ axis shows the descriptive beliefs about the disapproval group's opinions in the left panels and the belief about the approval group's opinions in the right panels. The $x$ axis shows the signal values for the participatory organization in panels A and B, and the ideological organization in panels C and D. For the ideological organization, we show the solution for an organization that aims to move expressed opinion toward approval. The brighter regions in the heatmap indicate higher rewards. We can see from the plot that the organizational signal has to align with the descriptive beliefs (brighter section near the diagonal). This is not surprising since this follows straightforwardly from the bounded confidence constraints. Due to those constraints on the generated signals, signals further away from the descriptive beliefs fetch lower rewards for the organization since the receivers ignore those signals.\par
Given this constraint, however, we can see the different strategies that ideological and participatory organizations pursue. Participatory organizations want to moderate beliefs about both in-group and out-group members: when players expect others to hold moderate beliefs, they also expect lower rhetoric. This encourages more to share their views, promoting the organization's participation goal. But ideological organizations want to move the mean of expressed opinions toward their preferred extreme (which is approval, in the example shown in Fig. \ref{fig:inst_rewards}). This requires inducing moderates in the approval group to remain silent. This is achieved by causing approval group members to believe that members of both groups hold more extreme views than they do and hence that they will engage in more extreme rhetoric; moderate approval group members do not get a high enough reward from expressing their opinions, in the face of opposing rhetoric and a substitution effect when their in-group engages in rhetoric, to warrant saying anything. The optimal signal when the ideological organization believes that the out-group's true belief is around $0.4$, for example, is around $0.2$ and the optimal signal when the in-group is believed to be at $0.7$ is close to $0.9$. It is also interesting to note that when the descriptive belief about the out-group is at the most extreme (close to zero), then the optimal signal for the ideological organization is to moderate that effect to prevent even extreme opinion holders from staying silent. We see this from the bright reward spot on the upper right corner in panel C.
Participatory organizations, on the other hand, want to make it seem that both the approval and disapproval groups are more moderate than they are. A participatory organization that believes the out-group is at $0.4$ will send signals close to $0.5$; if it believes the in-group is at $0.8$ it will send signals as close as possible to $0.6$.  That strategy ensures that the equilibrium rhetorical intensity of the group is lower and, therefore, it is optimal for even moderate opinion holders to express their opinions with a higher rhetorical intensity. \par 
Another important aspect to note here is that ideological organizations generally get higher rewards from signaling out-group opinions than in-group ones. (The brightest areas on the out-group heatmap correspond to higher values than the values for the brightest areas on in-group heatmap.) Meanwhile, for participatory organizations, there isn't a significant difference between out-group opinion and in-group opinion signaling. We can connect this observation to the empirical finding that in an ideological setting, content about political opponents is much more likely to be shared on social media \citep{rathje2021out}.
\begin{figure*}[!htbp]
\centering
         \includegraphics[width=\textwidth]{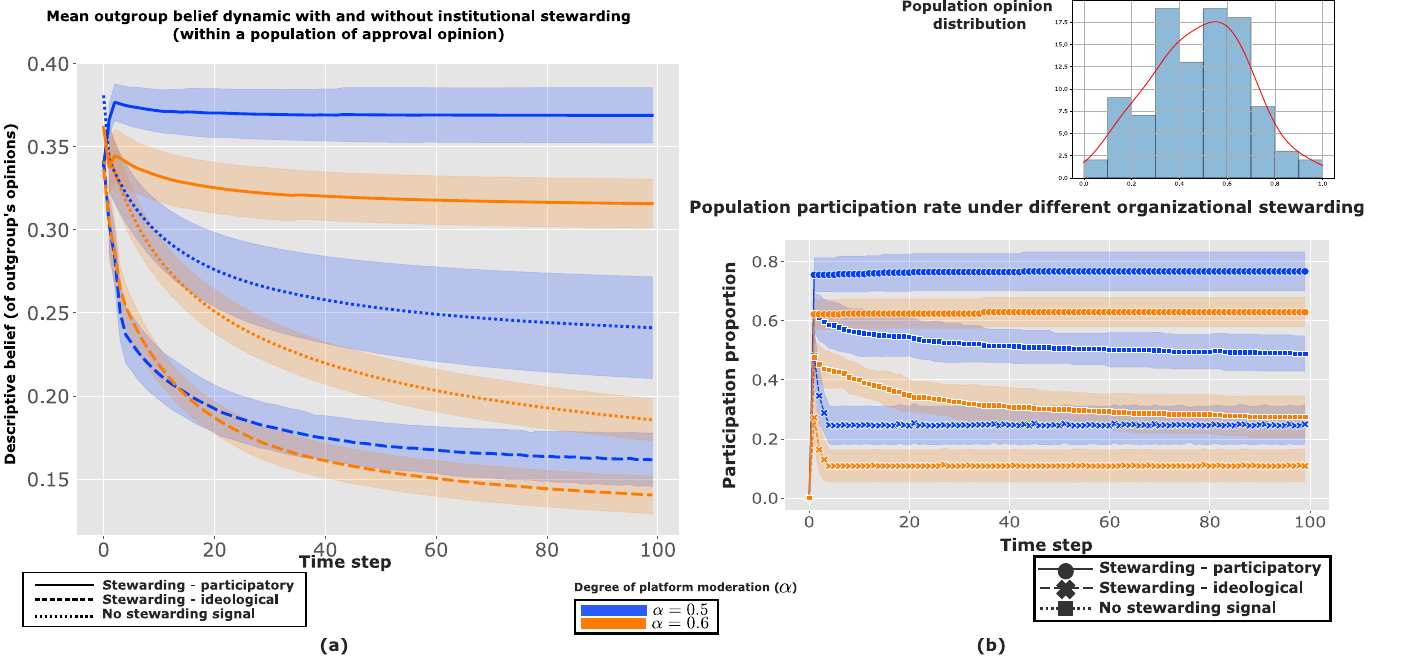}
\caption{Effect of organizational stewarding (participatory and ideological) on (a) descriptive beliefs about out-group's opinions and (b) rate of participation within a population. Shown here for the population of approval opinion holders. Stewarding under participatory and ideological organization is shown in continuous and dashed lines, respectively. Belief dynamic under no organizational stewarding is shown with dotted lines.}
\label{fig:opt_sig}
\end{figure*}
\section{Simulation of organizational stewarding}
\label{sec:long_run_dynamic}
Next, we analyze the long-run dynamics of the effects of organizational stewarding by ideological and participatory organizations using computational simulation. Each of the steps of the organization stewarding process described in Sec. \ref{sec:sub_inst_signal} is simulated as follows:
\begin{itemize}[wide, labelwidth=!, labelindent=0pt]
    \item \textit{Organizational signal generation:} We solve the optimization problem for ideological and participatory organizations as described in the previous section using the Value Iteration algorithm \citep{sutton2018reinforcement} which generates an optimal signaling policy for the two organizations.
    \item \textit{Belief update (from organization signal):} In order to calculate the posterior beliefs in Eqn. \ref{eqn:belief_update}, we use Monte Carlo estimation for Bayesian posteriors \citep{robert1999monte}.
    \item \textit{Opinion expression:} This step requires estimating the equilibrium rhetorical intensity from the perspective of each member of the population based on their own private type, that is, their opinion. As mentioned earlier, we use the equivalence between \emph{ex-ante} and \emph{ex-interim} Bayesian Nash equilibrium to simulate this process. Simply stated, the equivalence says that \emph{ex-interim} response averaged over types is the same as the \emph{ex-ante} response of the equilibrium (Lemma 6.1 in \citep{fujiwara2015bayesian}). For simulation, as a first step, this involves estimating the \emph{ex-ante} Bayesian Nash equilibrium of the opinion expression game played by two representative players with mean population opinion based on the prior beliefs about the opinion distribution. We calculate this by estimating the intersection of the Best Response curves of Eqns \ref{eqn:br_eqns_a} and \ref{eqn:br_eqns_b}. Since finding a closed-form solution for the curves is not feasible due to the transcendental nature of the equations, we use numerical root finding (Bisection) to approximate the intersection of the curves. In the second step, once this value is estimated, we simulate the individual rhetorical intensity response to the \emph{ex-ante} equilibrium by substituting it in Eqn. \ref{eqn:br_eqns_a} as the non-focal player's rhetorical intensity.  
    \item \textit{Belief update (from community signals):}  We model the beliefs about the approval and disapproval group in the population using a Beta distribution with parameters $a=5,b=3$ and $a=3,b=5$, respectively. Since the Bayesian update of a Beta distribution has a closed form solution, the updated distribution of the belief of the approval group's opinion is calculated as $Beta(a+o^{E}_{\geqslant 0.5},b+(1-o^{E}_{\geqslant 0.5}))$ where $o^{E}_{\geqslant 0.5}$ is the mean expressed opinion of the approval.
\end{itemize}

We run the simulation with a population of $N=100$ agents with opinions drawn from a bimodal Gaussian distribution with parameters $\mu_{1}=0.4,\mu_{2}=0.6,\sigma_{\{1,2\}}=0.2$, mixture coefficient 0.5. The rhetorical intensity threshold at which an individual stays silent ($\gamma^{\text{silence}}$) is set at 0.3. \par
We run the simulations under homogeneous beliefs, which means that within a group (approval or disapproval), all agents share the same beliefs about the population's mean approval and disapproval opinions. We run 10 batches of simulations for 100 timesteps. We also run three sets of simulations, with each set corresponding to the dynamic under the participatory, ideological, and \emph{no organizational} stewarding. For the \emph{no organizational} stewarding run, we include only steps 3 and 4 of the organizational stewarding process, eliminating the organizational signals completely and updating the beliefs solely based on the community interactions.\par
Fig. \ref{fig:opt_sig} shows the plot of the mean and standard deviation of the following set of attributes of the population calculated across the batches of run: i) descriptive belief about the disapproval group's opinion from the perspective of the approval group in Fig. \ref{fig:opt_sig}a; ii) participation rate (measured as the proportion of the population who express their opinion, either of approval or disapproval) in Fig. \ref{fig:opt_sig}b. The opinion distribution from which the simulation was run is shown in the inset of the Figure. We also repeat the simulation for different values of $\alpha = 0.5, 0.6$, which models different degrees of leniency in platform moderation policy. The dynamic of the descriptive beliefs and participation rate under participatory, ideological, organization is shown with a smooth and dashed line, respectively. The dynamic without any kind of organization stewarding is shown in a dotted line. 
Although we show the plots for the beliefs of the approval group in Fig. \ref{fig:opt_sig} a, the beliefs of the disapproval group have identical patterns, but with the range of opinion values flipped for the in-group and out-group beliefs.\par
Based on the simulations, we see a general trend that participation under the stewarding of participatory organizations remains stable at high levels. The participation rates are much lower without any organizational stewarding and under an ideological organization. This effect is exacerbated by a stricter platform moderation policy too. This is because, under a strict moderation policy uniformly applied to the entire population irrespective of the particular opinion, the moderate opinion holders are more likely to stay silent as moderation costs increase, given the lower utility they enjoy from expressing their opinion. This result provides a key insight: even though a more lenient moderation practice may improve participation (participation rate is higher for $\alpha=0.5$ compared to $\alpha=0.6$); there is a confounding effect from the type of organizational stewarding. In other words, the increased participation can be offset if the community is stewarded by an ideological organization. We can see this reflected in participation rate under ideological stewarding with $\alpha=0.5$ compared to participatory stewarding with $\alpha=0.6$. We also compare the participation rate without any stewarding and observe that ideological stewarding is worse than under no organizational stewarding at all. \par 
With respect to the beliefs about the out-group opinions shown in Fig. \ref{fig:opt_sig}a, we see that stewarding under ideological organizations also produces the most distorting effect on the beliefs, where the population, over time, believes that the out-group is more extreme than it really is. This effect is primarily driven by the differences in participation between the moderate and extreme opinion holders. Meanwhile, under participatory organization stewarding, the distortion of beliefs is minimal. \par
We also note the effect of a moderation policy applied uniformly for all opinion holders. When such a policy is uniformly stricter for everyone, the distortion of belief becomes worse. This is because the increased cost affects moderated more, resulting in higher non-participation of the moderate opinion holders and thereby results in greater distortion of beliefs. This is an important insight into the role of platform moderation: simply adjusting the moderation levels in a coarse grained way might not lead to improvements in either participation or distortion of second order beliefs about the out-group.\par 

\begin{figure*}[t]
\centering
         \includegraphics[width=0.75\textwidth]{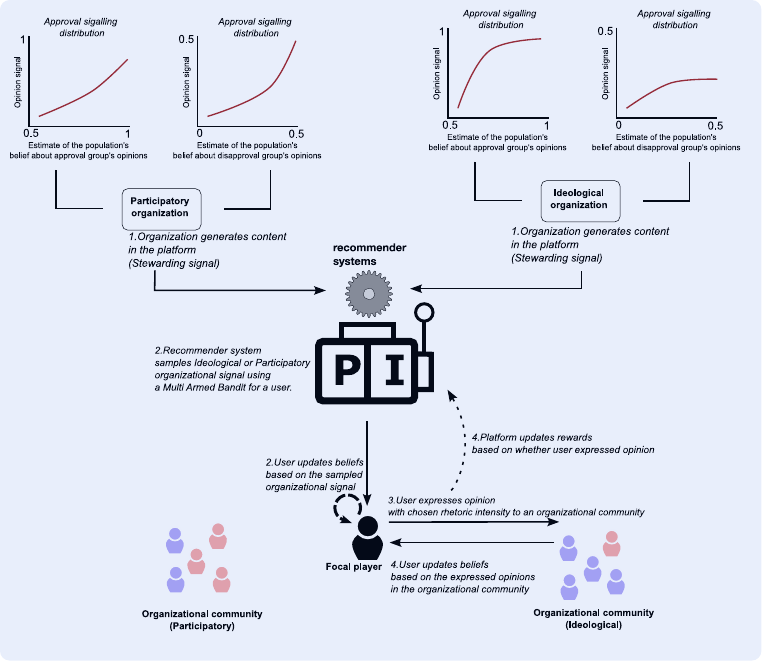}
\caption{Extended schematic representation of the organizational stewarding process in Fig. \ref{fig:single_inst_stewarding_diag} with the addition of participatory and ideological organizations, recommender systems, and organizational communities.}
\label{fig:institutional_selection}
\end{figure*}
\section{Platforms and organizational communities}
\label{sec:multi_inst}
While the analysis of a single organization within a population with homogeneous beliefs offers valuable insights, the reality of online platforms is much more complex. On platforms like Reddit, Twitter, and YouTube, digital communities are formed around a central theme, such as subreddits on Reddit, influential accounts on Twitter (X), and channels on YouTube.
In many of these digital communities, members share, discuss, and enforce their perceptions of right and wrong, often along a partisan axis \citep{conover2011political, waller2021quantifying}. This speaks to the presence of a diverse population with heterogeneous beliefs and multiple viewpoint organizations, each stewarding the beliefs of the community that falls within its scope. In this section, using the organizational stewarding process presented earlier, we develop a model of organizational community formation for both participatory and ideological organizations within a platform. Furthermore, we analyze the characteristics of the beliefs of each organizational community and show how ideological organizational communities give rise to the phenomenon of false polarization, i.e., individuals in ideological communities hold more extreme opinions and, at the same time, perceive the out-group to be more extreme than in reality \citep{levendusky2016mis}. \par

\subsection{Organizational communities} For most online platforms, it is common to use a recommendation algorithm that connects users with communities or accounts that match their interests and engagement patterns. Although there are several differences between platforms in terms of how these processes are implemented in practice, we use a minimal set of key entities common to most platforms: organizations, users, and the platform. \par 

The role of organizations and users is modeled as presented earlier in the organizational stewarding model. In reality, these organizations are often channels, subreddits, accounts etc. that generate information signals about approval and disapproval groups on normative issues. In this section, we use three categories of organizations in our simulation of the population under multiple organizations: one participatory organization, one ideological organization that has the objective of moving the expressed opinion towards approval, and one ideological organization that has the objective of moving the expressed opinion towards disapproval. As before, users within the population are modeled as receivers of organizational signals and they express opinion within the digital space provided by the platform. \\

\subsection{Modeling the platform} 
The platform is an intermediary between the organizations and the user, and it is responsible for the following set of functions in our model: \\
i. \textit{Content recommendation} -- A recommender system that recommends an organization and the organizational signal to a user based on their past interactions with the organization. In practice, a platform often uses a combination of various objectives, including engagement and diversity, to fine-tune the recommender system for a given user \citep{stray2022building}, however, in our model, the goal of the recommender system is to recommend an organization to a user to maximize the probability that the user expresses their opinion on the platform. In terms of technology, Bandit algorithms have become one of the mainstays of recommender systems \citep{silva2022multi, parapar2021diverse}. With that motivation, we use a vanilla UCB algorithm \citep{auer2002finite} to simulate the platform's recommender engine, which recommends a participatory or ideological organizational signal to a user. \\
ii. \textit{Digital space } -- The platform maintains a digital space for each organization where users react to organizational signals by expressing their opinions and interacting with others in the community. One can think of this digital space as the comment and user interaction sections available in most platforms. We refer to this digital space for an organization as the organizational community.\\
iii. \textit{Moderation} -- The platforms is also responsible for moderating the content within the digital space. In our case, this is modeled as the platform's control of the cost parameter $\alpha$. \\

With the main components of the models in place, we present the sequence of steps involved in organizational community formation. The main idea behind the process is that the user's interaction with the content generated by a particular organization is visible and registered with the recommender system when the user expresses their opinion based on the organizational signal. The recommender system, in turn, learns to recommend the particular organizations based on the user's past participation and the predicted attributes of the user, in this case, the beliefs about the in-group and out-group. The objective of the platform is to maximize user participation within the platform. This is done by selecting the right organization for the user and ensuring that the signals it recommends from the selected organization induce the user to participate by expressing their opinion. Since participatory organizations generate signals that allow moderate opinion holders to participate in their community, and ideological organizational signals suppress moderate opinion participation, we see a higher representation of moderate opinion holders in participatory organizational communities. We elaborate these steps in more detail below.\par
\noindent \emph{Initial setup:} Various participatory and ideological organizations in the platform generate signals based on their optimal signaling strategy. On a given issue, due to the diversity of the content and organizations in the platform, we assume that there is content available for the recommender system to present to a user from a wide spectrum of opinions for both approval and disapproval. For every user, the platform maintains a history of past interaction with two types of organization: the participatory organization and the ideological organization whose goal is to move the opinion in the direction that aligns with the user's opinion. We outline the modeling steps in more detail below.  \par
\noindent At each time step $t$:
\begin{enumerate}[wide, labelindent=0pt]
    \item For each user $i$, the platform's recommender system samples a participatory or an ideological organizational signal by pulling one of the two arms of the Bandit based on the UCB values corresponding to each arm.
    \item The user receives the signal from the organization sampled by the recommender system. After a particular organization is chosen by the system, there is still the question of the type of the information signal: whether the signal is about the approval or the disapproval group's opinion since each organization maintain two separate signaling distributions. We model this choice as a stochastic policy where each type of information (approval or disapproval group's opinion) is selected in proportion to the maximum reward the organization can achieve from either of the two distributions. We model this choice based on the assumption that the overall proportion of information content about the two groups' opinions generated by the organization will be proportional to the rewards that each type of information fetches the organization. Once the user receives the signal from the organization, they update their belief based on the signal received. \par
\item The user makes the choice to either \emph{express} their opinion or \emph{stay silent}. This decision is determined by the equilibrium rhetorical intensity discussed in step \emph{opinion expression} in Sec. \ref{sec:sub_inst_signal}. Specifically, the user expresses if their best-response rhetorical intensity is higher than a threshold $\gamma^{\text{silence}}=0.3$.
\item Expression of the user's opinion is recorded in the digital space provided by the platform for the corresponding organizations. Each user who expresses their opinion in the digital space is added to the organizational community maintained by the platform.
\item Each user who was added to the organizational community observes the opinion of other users within the community and updates their corresponding descriptive beliefs. For the participatory institutional communities, these community signals consist of both approval and disapproval opinions since both types of users interact in the participatory institutions. Whereas, for the ideological communities, since only one of the groups interacts in their corresponding community, the community signal is only either approval or disapproval opinions.  
    \item The recommender system updates the rewards for the corresponding arm to +1 if the user expressed based on the sampled signal or 0 if they did not express.

\end{enumerate}
\begin{table*}[htbp]
\centering
\begin{tabular}{p{2cm}p{1cm}p{4cm}p{4cm}p{4cm}}
\toprule
\textbf{Community Category} & \textbf{Opinion Group} & \textbf{Opinion (Mean ± SD)} & \textbf{Belief about out-group (Mean ± SD)} & \textbf{Belief about in-group (Mean ± SD)} \\
\midrule
\multirow{1}{*}{Participatory} & $\geq0.5$ & \(0.666 \pm 0.114\) & \(0.378 \pm 0.066\) & \(0.5 \pm 0.001\) \\
\midrule
\multirow{1}{*}{Ideological}   & $\geq0.5$ & \(0.703 \pm 0.104\) & \(0.180 \pm 0.093\) & \(0.707 \pm 0.002\) \\
\midrule
\multirow{1}{*}{Silent}        & $\geq0.5$ & \(0.516 \pm 0.013\) & \(0.161 \pm 0.134\) & \(0.907 \pm 0.11\) \\
\bottomrule
\end{tabular}
\caption{Summary of opinion and belief about out-group and in-group for participatory and ideological communities along with the characteristics of agents that do not participate in any community. The true population mean opinion is 0.4 for disapproval and 0.6 for approval.}
\label{tab:opinion_and_belief_diff}
\end{table*}

\begin{table*}[htbp]
\centering
\begin{tabular}{lp{1cm}p{4cm}p{4cm}p{4cm}}
\toprule
\textbf{Comparison} & \textbf{Opinion Group} & \textbf{Cohen's \(d\) (Opinion)} & \textbf{Cohen's \(d\) (Out Belief)} & \textbf{Cohen's \(d\) (In Belief)} \\
\midrule
Ideological vs Participatory & \(\geq0.5\) & \(0.338^{*}\) & \(-2.488^{***}\) & \(100.459^{***}\) \\
Ideological vs Silent & \(\geq0.5\) & \(1.852^{***}\) & \(0.201^{*}\) & \(-6.732^{***}\) \\
Participatory vs Silent & \(\geq0.5\) & \(1.345^{**}\) & \(3.014^{***}\) & \(-15.548^{***}\) \\
\bottomrule
\end{tabular}
\caption{Results of Cohen's d estimate for effect size for comparison of opinion, out-group belief, and in-group belief for each organizational community. $*$: small to medium, $**$: medium to large effect size, $***$: large to very large effect size.}
\label{tab:cohens_d}
\end{table*}

\subsection{Organizational community characteristics} 
We can analyze the characteristics of the organizational communities that form based on the combination of the organization signaling, user expression of opinion, and recommender systems described above. Specifically, we answer the following two questions about the characteristics of the community:
\begin{itemize}[wide, labelindent=0pt]
    \item{ Opinion differences: \textit{How do the opinions differ among the members in each organizational community?}}
    \item{ Second-order beliefs about out-group opinions: \textit{How do the beliefs about the out-group differ among the members in each organizational community?}} \item{ Second-order beliefs about in-group opinions: \textit{How do the beliefs about the in-group differ among the members in each organizational community?}}
\end{itemize}
\noindent \textbf{Simulation setup:} Similar to Sec. \ref{sec:long_run_dynamic}, to answer the above questions, we run simulations with the multiple institutions setup and heterogeneous beliefs, that is, users hold different beliefs about their in-group and out-group sub-population. The population size, opinion distribution are same as in simulation in Sec. \ref{sec:long_run_dynamic}), i.e, population of $N=100$ for $T=100$ time steps with opinions drawn from a bimodal Gaussian distribution with parameters $\mu_{1}=0.4,\mu_{2}=0.6,\sigma_{\{1,2\}}=0.2$, mixture coefficient 0.5. We model the heterogeneous beliefs about the approval and disapproval group in the population using two randomly sampled Beta distribution for each user (corresponding to the belief about in-group and out-group) such that the mean value of the parameters $a$ and $b$ for the beliefs of approval and disapproval for the whole population are $a=5,b=3$ and $a=3,b=5$, respectively. The threshold at which an individual stays silent is 0.3. We answer the questions about the characteristics of the community in the second half of the simulation run in order for the recommender system to stabilize the learning for each user.\par
Table \ref{tab:opinion_and_belief_diff} shows the difference in opinion distribution for the participatory and ideological communities. There is a statistically significant difference in the distribution of opinions between the two communities, with the mean opinion of the ideological community more extreme than that of participatory communities. This difference results directly from higher likelihood of moderate opinion holders staying silent from signals of ideological organizations, and consequently, the recommender systems matching them to participatory organizations more often than extreme opinion holders. Next, we look at the differences in out-group beliefs of the two communities. We select the approval group for analysis, although the characteristics are the same for both groups. We see a significant difference between the two communities with respect to out-group and in-group beliefs; beliefs about both in-group and out-group are distorted towards the extreme for ideological communities. However, for participatory communities, the belief about in-group is more moderate than the true opinion. This relationship results from a combination of the optimal signaling schemes constructed by organizations and the recommender system that stratifies users to different communities. As extreme opinion holders participate more in ideological communities, their out-group belief also drift towards one extreme, a dynamic we observed in Sec. \ref{sec:long_run_dynamic}). \par
Finally, our analysis also sheds light on the characteristics of the population that does not participate in either of the communities. This population consists of users with moderate opinions and more extreme beliefs about others' opinions. This is because more moderate opinion holders with more extreme beliefs about the population are likelier to stay silent. Although the recommender system samples different organizations, it fails to bring those users to participate in either organization's community.

\section{Conclusion and discussion}
In this paper, we present a model that shows that differences in a population in the expression of opinions can cause polarization at the population and community levels, both in the opinions expressed and beliefs about others. The observed phenomenon arises from the differences in individual incentives for opinion expression, organizational incentives, and the recommender system of online platforms. As we have shown, polarization can arise not from a distortion in actual opinions but rather from rational choices about when and how to express opinions. Rhetorical intensity is shown to pay a key role: rhetoric generates payoffs for individuals with strong views but also inhibits expression by those with moderate views. If a platform attempts to moderate intense rhetoric, however, this may only exacerbate the problem: moderates will be less willing to incur the costs of rhetoric than those with extreme views.  

Importantly, we show these effects without assuming any change in the actual opinions held by individuals, only in the pattern of expression of opinions. This result highlights a risk from AI that has not been adequately appreciated: reliance on the expression of opinions on AI-powered social media platforms can distort perceptions in public debates and policy-making about the true distribution of views. That in turn can distort actual policy-making.  Additionally, distortions in the expression of opinion also distort the data on which AI models are trained, amplifying the effect of a shift to AI-mediated interaction. Moreover, correction for the distorting effect is difficult: there are many reasons that users do not participate in particular social media discussions and so the signal from silence is difficult to extract. 

The role we have identified for platform AI-based recommender systems opens up the possibility of designing mitigation strategies that do not require changing underlying opinions. We identify at least two such strategies for online platforms as follows:\\
\noindent \emph{A tailored approach to content moderation:} In our model, the difference in rhetorical intensity arises from the difference in utility each individual gains from opinion expression. We assumed platform moderation of rhetorical intensity imposes costs ($\alpha$) that are invariant with respect to the individual's opinion. If online platforms were able, however, to customize content moderation policy to allows more leniency (lower $\alpha$) in the expression of the opinion of moderate opinion holders, the resulting difference in optimal rhetorical intensity between extreme and moderate opinion holders could be corrected.\par
\noindent \emph{Prioritizing signals from participatory organizations:} In conjunction with a tailored moderation policy, platforms could also tailor their recommender strategy to prioritize signals from participatory organizations over ideological organizations for extreme opinion holders. Since extreme opinion holders still express their opinion even under signals from participatory organizations, a higher proportion of participatory organizational signals can bring their beliefs about out-group opinions more in line with less polarized participatory organizational communities.
\par Both of these strategies require, however, that platforms engage in a form of content-based regulation of rhetoric. People who oppose animal-testing in any context would be more restricted in their platform speech than those who believe animal-testing is acceptable in some contexts but not others. That may itself be unacceptable for political communities. The only alternative would seem to be aggressive regulation of rhetoric across the board. In a sense, this is what political communities have historically sought to achieve, through strong rules of civility in political forums (such as legislatures) and strong editorial norms in publications such as widely-distributed newspapers. The challenge in the era of AI-based opinion platforms is to establish such norms in a highly decentralized setting. Our results emphasize that the benefits to developing such norms is important not merely to support individual expression rights but also to ensure the integrity of public and policy-makers perceptions of true public opinion on contentious policy matters and the integrity of the data on which our large language models are trained.
\section{Acknowledgments}
We thank the following people for their feedback on this work: Graham Noblit, Valerie Platsko, Kathryn E. Spier, Peter Marbach, Ashton Anderson.
\bibliographystyle{ACM-Reference-Format}
\bibliography{sample-base}

\newpage
\onecolumn
\newgeometry{left=4cm,right=4cm,marginparwidth=3.5cm,marginparsep=0.3cm}
\edef\marginnotetextwidth{\the\textwidth}
\renewcommand\marginfont{%
    \normalfont\scriptsize\itshape
}
\section{Appendix}

\subsection{Proof of Theorem \ref{theo:exp_bound}}

\noindent \textit{Theorem statement:} Let $v(o_{i})=o_{i}$ if $o_{i} \geqslant 0.5$ and $v(o_{i})=1-o_{i}$, if $o_{i} < 0.5$ and $o_{i} \sim f$ is drawn from any arbitrary prior opinion distribution with \texttt{p.d.f} $f$ and \texttt{c.d.f} F, and $\hat{v}_{o} = \mathrm{E}_{f}[v(o_{i})]$. Then, under the condition of equal support of approval and disapproval, i.e., $F(0.5)=0.5$, there exists an ex-ante symmetric Bayesian Nash equilibrium $\gamma_{i}^{*}=0$, $\forall i$, if an only if $\hat{v}_{o} < \frac{\alpha}{1-\hat{n}\cdot (1-\lambda_{in})}$  

\begin{proof}
\noindent \textit{Sketch:} The proof follows three steps: (i) we calculate the expected utility in ex-ante form for the players (Lemma \ref{lemma:exp_util}), (ii) we calculate the best response functions based on the ex-ante utilities, and (iii) we show that the symmetric intersection of the best response function at $\gamma=0$ exists only under the condition stated in the theorem.\\
\begin{lemma}
    The ex-ante utility of a representative player in a population with opinion distribution $o \sim f$ is 
    \begin{equation*}
        \mathrm{E}_{f}[u_{i}(\gamma_{i},\gamma_{-i})] 
= \hat{n} \cdot \mathrm{E}_{f}[v(o)] \cdot \lambda_{\text{in.}} \cdot \gamma_{i}^{(1-\gamma_{-i})}+ (1-\hat{n}) \cdot  \mathrm{E}_{f}[v(o)]  \cdot \lambda_{\text{out}}^{\gamma_{-i}} \cdot \gamma_{i} - \alpha\cdot \gamma_{i}
    \end{equation*}
    \label{lemma:exp_util}
\end{lemma}

The ex-ante utility is based on the expected utility of a random player in the game with opinion ($o \sim f$) drawn from an opinion distribution $f$. This utility can be calculated by taking expectation based on Eqn \ref{eqn:utility} and Eqn \ref{eqn:utility_disappr} as 
\begin{align}
\mathrm{E}_{f}[u_{i}(\gamma_{i},\gamma_{-i};o_{i})] &= \mathrm{E}_{f}[u_{i}(\gamma_{i},\gamma_{-i};o_{i}|o_{i} < 0.5)] +
\mathrm{E}_{f}[u_{i}(\gamma_{i},\gamma_{-i};o_{i}|o_{i} \geqslant 0.5)] \\
&= \int_{0}^{0.5} [(1-\hat{n}) \cdot (1-o) \cdot \lambda_{\text{in.}} \cdot \gamma_{i}^{(1-\gamma_{-i})}+ \hat{n} \cdot  (1-o)  \cdot \lambda_{\text{out}}^{\gamma_{-i}} \cdot \gamma_{i} - \alpha\cdot \gamma_{i} ] \cdot f(o) \, do \marginnote{Expanding the equations}\\
&\phantom{{}=1}+ \int_{0.5}^{1} [\hat{n} \cdot o \cdot \lambda_{\text{in.}} \cdot \gamma_{i}^{(1-\gamma_{-i})}+ (1-\hat{n}) \cdot  o  \cdot \lambda_{\text{out}}^{\gamma_{-i}} \cdot \gamma_{i} - \alpha\cdot \gamma_{i} ] \cdot f(o) \, do
\end{align}
Replacing $v(o)=o$ if $o \geqslant 0.5$ and $v(o)=1-o$, if $o < 0.5$ in the above equation, the expected utility becomes
\begin{align*}
\mathrm{E}_{f}[u_{i}(\gamma_{i},\gamma_{-i};o)] 
&= \int_{0}^{0.5} [(1-\hat{n}) \cdot v(o) \cdot \lambda_{\text{in.}} \cdot \gamma_{i}^{(1-\gamma_{-i})}+ \hat{n} \cdot  v(o)  \cdot \lambda_{\text{out}}^{\gamma_{-i}} \cdot \gamma_{i} - \alpha\cdot \gamma_{i} ] \cdot f(o) \, do \\
&\phantom{{}=1}+ \int_{0.5}^{1} [\hat{n} \cdot v(o) \cdot \lambda_{\text{in.}} \cdot \gamma_{i}^{(1-\gamma_{-i})}+ (1-\hat{n}) \cdot  v(o)  \cdot \lambda_{\text{out}}^{\gamma_{-i}} \cdot \gamma_{i} - \alpha\cdot \gamma_{i} ] \cdot f(o) \, do\\
&= ((1-\hat{n}) \cdot \lambda_{\text{in.}} \cdot \gamma_{i}^{(1-\gamma_{-i})}) \int_{0}^{0.5}  v(o)  \cdot f(o) \, do + (\hat{n} \cdot \lambda_{\text{in.}} \cdot \gamma_{i}^{(1-\gamma_{-i})}) \int_{0.5}^{1}  v(o)  \cdot f(o) \, do\\
&\phantom{{}=1}+ (\hat{n} \cdot \lambda_{\text{out}}^{\gamma_{-i}} \cdot \gamma_{i}) \int_{0}^{0.5}  v(o)  \cdot f(o) \, do + ((1-\hat{n}) \cdot \lambda_{\text{out}}^{\gamma_{-i}} \cdot \gamma_{i}) \int_{0.5}^{1}  v(o)  \cdot f(o) \, do \\
&\phantom{{}=2} - \alpha\cdot \gamma_{i}\int_{0}^{1} f(o) \, do \marginnote{Rearraging the terms}\\
\end{align*}
\begin{align*}
&= (\hat{n} \cdot \lambda_{\text{in.}} \cdot \gamma_{i}^{(1-\gamma_{-i})}) [\int_{0}^{0.5}  v(o)  \cdot f(o) \, do + \int_{0.5}^{1}  v(o)  \cdot f(o) \, do]\\
&\phantom{{}=1}+ ((1-\hat{n}) \cdot \lambda_{\text{out}}^{\gamma_{-i}} \cdot \gamma_{i}) [\int_{0}^{0.5}  v(o)  \cdot f(o) \, do + \int_{0.5}^{1}  v(o)  \cdot f(o) \, do] \marginnote{Since $\hat{n}=1-\hat{n}$ when $F(0.5)=0.5$}\\
&\phantom{{}=2} - \alpha\cdot \gamma_{i}  \marginnote{and $\int_{0}^{1} f(o) \, do = 1$}\\
\end{align*}
Substituting $\mathrm{E}_{f}[v(o)] = \int_{0}^{0.5}  v(o)  \cdot f(o) \, do + \int_{0.5}^{1}  v(o)  \cdot f(o) \, do$ in the above equation and rearranging the terms
\begin{align*}
\mathrm{E}_{f}[u_{i}(\gamma_{i},\gamma_{-i};o)] 
&= \hat{n} \cdot \mathrm{E}_{f}[v(o)] \cdot \lambda_{\text{in.}} \cdot \gamma_{i}^{(1-\gamma_{-i})}+ (1-\hat{n}) \cdot  \mathrm{E}_{f}[v(o)]  \cdot \lambda_{\text{out}}^{\gamma_{-i}} \cdot \gamma_{i} - \alpha\cdot \gamma_{i}
\end{align*}
\end{proof}
Note that the above equation of the utility in ex-ante form has the same structure as the utility in Eqn. \ref{eqn:utility}. Taking the derivative of the above function with respect to the strategic variable $\gamma_{i}$ and $\gamma_{-i}$ and solving the two set of equations:
\begin{align*}
   \frac{\partial}{\partial \gamma_{i}} \mathrm{E}_{f}[u_{i}(\gamma_{i},\gamma_{-i};o)] = 0 \, \,   \text{and} \, \, \frac{\partial}{\partial \gamma_{-i}} \mathrm{E}_{f}[u_{i}(\gamma_{i},\gamma_{-i};o)] = 0
\end{align*}
yields the two best response functions in ex-ante form as follows.
\begin{equation}
    BR_{i}(\gamma_{-i};o) = \min \left( 1, \max \left( 0, \left( \frac{ \hat{n} \cdot \mathrm{E}_{f}[v(o)] \cdot \lambda_{\text{in}} \cdot (1-\gamma_{-i})}{\alpha - (1-\hat{n})(\mathrm{E}_{f}[v(o)]\lambda_{\text{out}}^{\gamma_{-i}})} \right)^{\frac{1}{\gamma_{-i,}}}\right) \right)
    \label{eqn:ex_ante_br_a}
\end{equation}
\begin{equation}
    BR_{-i}(\gamma_{i};o_{-i}) = \min \left( 1, \max \left( 0, \left( \frac{ \hat{n} \cdot \mathrm{E}_{f}[v(o)] \cdot \lambda_{\text{in}} \cdot (1-\gamma_{i})}{\alpha - (1-\hat{n})(\mathrm{E}_{f}[v(o)]\lambda_{\text{out}}^{\gamma_{i}})} \right)^{\frac{1}{\gamma_{i,}}} \right) \right)   
    \label{eqn:ex_ante_br_b}
\end{equation}

For an ex-ante symmetric Bayesian Nash equilibrium $\gamma_{i}^{*}=0$, $\forall i$, to exist, the following condition needs to hold true.
$
BR_{i}(\gamma_{-i}=0;o_{i})=0
$
and 
$
BR_{-i}(\gamma_{i}=0;o_{-i})=0
$.
In the above best response function Eqn. \ref{eqn:ex_ante_br_a}, when $\gamma_{-i}=0$, the required equilibrium condition $BR_{i}(\gamma_{-i}=0;o_{i})=0$ holds when the numerator in the expression is less than the denominator. That is, the following condition has to be true:
\begin{align}
   \hat{n} \cdot \mathrm{E}_{f}[v(o)] \cdot \lambda_{\text{in}} <& \alpha - (1-\hat{n})(\mathrm{E}_{f}[v(o)]\lambda_{\text{out}}^{\gamma_{-i}} ) \\
   \mathrm{E}_{f}[v(o)] <& \frac{\alpha}{1-\hat{n}\cdot (1-\lambda_{in})} \marginnote{Rearraging the terms}
   \label{eqn:appdx_cond}
\end{align}
The above condition also follows from the second condition of the $\gamma_{i}^{*}=0$ equilibrium existence, that is, $BR_{-i}(\gamma_{i}=0;o_{-i})=0$. Therefore, we establish the sufficiency condition that if there is an ex-ante symmetric Bayesian Nash equilibrium $\gamma_{i}^{*}=0$, $\forall i$, the condition $\mathrm{E}_{f}[v(o)] < \frac{\alpha}{1-\hat{n}\cdot (1-\lambda_{in})}$ holds true. \par
Similarly, from the condition in Eqn \ref{eqn:appdx_cond} and when i. $\gamma_{-i} \rightarrow 0$, the following is true $BR_{i}(\gamma_{-i};o_{i})=0$ in Eqn \ref{eqn:ex_ante_br_a}, and ii. $\gamma_{i} \rightarrow 0$, the following is true $BR_{-i}(\gamma_{i};o_{-i})=0$ in Eqn \ref{eqn:ex_ante_br_b}. These two are the conditions for the equilibrium at $\gamma = 0$, and therefore establishes the necessary condition that if $\mathrm{E}_{f}[v(o)] < \frac{\alpha}{1-\hat{n}\cdot (1-\lambda_{in})}$ holds, there is an ex-ante symmetric Bayesian Nash equilibrium $\gamma^{*}=0$.

\subsection{Simulation code}

Code for running the simulations included in the paper can be found under \texttt{\url{https://github.com/atrisha/normative_stewarding}}
\restoregeometry

\end{document}